\documentclass[runningheads]{llncs}

\usepackage{amsmath}
\usepackage{amssymb}
\usepackage{theorem}
\usepackage{graphicx}
\usepackage{hyperref}
\usepackage{tikz}
\usepackage{verbatim}
\usepackage{todonotes}
\usepackage[capitalise]{cleveref}
\usepackage{cite}

\newcommand{\MAW}{\text{MAW}}
\newcommand{\ST}{\textit{ST}}

\newcommand{\nodestring}{\textsf{str}}
\newcommand{\RT}{\textsf{T}}
\newcommand{\UT}{\textsf{T}}

\usetikzlibrary{arrows,decorations.pathreplacing}
\usepackage{xspace}

\begin{document}

\title{Minimal Absent Words in Rooted and Unrooted Trees}

\author{
Gabriele Fici\inst{1}   \and Pawe{\l}  Gawrychowski\inst{2}
}

\institute{
Dipartimento di Matematica e Informatica, Universit\`a di Palermo, Italy\\ \email{gabriele.fici@unipa.it} 
\and
Institute of Computer Science, University of Wroc{\l}aw, Poland\\ 
\email{gawry@cs.uni.wroc.pl}
}

\maketitle

\begin{abstract}
We extend the theory of minimal absent words to (rooted and unrooted) trees, having edges labeled by letters from an alphabet $\Sigma$ of cardinality $\sigma$. We show that the set $\MAW(T)$ of minimal absent words of a rooted (resp.~unrooted) tree $T$ with $n$ nodes has cardinality $O(n\sigma)$ (resp.~$O(n^{2}\sigma)$), and we show that these bounds are realized. Then, we exhibit algorithms to compute all minimal absent words in a rooted (resp.~unrooted) tree in output-sensitive time $O(n+|\MAW(T)|)$ (resp.~$O(n^{2}+|\MAW(T)|)$ assuming an integer alphabet of size polynomial in $n$.
\end{abstract}

\section{Introduction}

Minimal absent words (a.k.a.~minimal forbidden words or minimal forbidden factors) are a useful combinatorial tool for investigating words (strings). A word $u$ is absent from a word $w$ if $u$ does not occur (as a factor) in $w$, and it is minimal if all its proper factors occur in $w$. This definition naturally extends to languages of words closed under taking factors.

The theory of minimal absent words has been developed in a series of papers~\cite{DBLP:conf/stacs/BealMR96,DBLP:journals/ipl/CrochemoreMR98,DBLP:conf/birthday/MignosiRS99,DBLP:journals/tcs/MignosiRS02,DBLP:journals/fuin/BealCM03} (the reader is pointed to~\cite{fici} for a survey on these results).
Minimal absent words  have then found applications in several areas, e.g., data compression~\cite{DCA,DBLP:conf/sccc/CrochemoreN02,DBLP:conf/dcc/FialaH08,DBLP:conf/isita/OtaM10}, on-line pattern matching~\cite{DBLP:conf/fct/CrochemoreHKMPR17}, sequence comparison~\cite{Charalampopoulos2018,SPIRE2018}, sequence assembly~\cite{FICI2006214,DBLP:journals/ita/MignosiRS01},  bioinformatics~\cite{Chairungsee2012109,DBLP:journals/bioinformatics/SilvaPCPF15,FiLaLoRi18}, musical data extraction~\cite{DBLP:conf/ismir/CrawfordB018}.

Bounds on the number of minimal absent words have been extensively investigated. The upper bound on the number of minimal absent words of a word of length $n$ over an alphabet of size $\sigma$ is $O(n\sigma)$~\cite{DBLP:journals/ipl/CrochemoreMR98,DBLP:journals/tcs/MignosiRS02}, and this is tight for integer alphabets~\cite{Charalampopoulos2018}; in fact, for large alphabets, such as when $\sigma\geq\sqrt{n}$, this bound is also tight even for minimal absent words having the same length~\cite{DBLP:journals/almob/AlmirantisCGIMP17}. 

Several algorithms are known to compute the set of minimal absent words of a word.  State-of-the-art algorithms compute all minimal absent words of a word of length $n$  over an alphabet of size $\sigma$ in time $O(n\sigma)$~\cite{DBLP:journals/ipl/CrochemoreMR98,DBLP:journals/bmcbi/BartonHMP14} or in output-sensitive $O(n + |\MAW(w)|)$ time~\cite{DBLP:conf/mfcs/FujishigeTIBT16,SPIRE2018} for integer alphabets. Space-efficient data structures based on the Burrows-Wheeler transform can also be applied for this computation~\cite{DBLP:conf/esa/BelazzouguiCKM13,DBLP:journals/algorithmica/BelazzouguiC17}. 

For a finite set of words $P$ over an alphabet of size $\sigma$, the minimal absent words of the factorial closure of $P$ can be computed in $O(|P|^{2}\sigma)$~\cite{DBLP:journals/fuin/BealCM03}, where $|P|$ is the sum of the lengths of the words of $P$.
Generalizations of minimal absent words have been considered for circular words~\cite{Charalampopoulos2018,FiReRi19} and multi-dimensional shifts~\cite{DBLP:journals/ijac/BealFM05}.

In this paper, we extend the theory of minimal absent words to trees. We consider trees with edges labeled by letters from an integer
alphabet $\Sigma$ of cardinality $\sigma$ polynomial in $n$.
In the case of a rooted tree $\RT$, every node $v$ is associated with a word $\nodestring(v)$, defined as the sequence of edge labels from $v$ to the root. A rooted tree $\RT$ can therefore be seen as a set  of words $L_{\RT}=\{\nodestring(v)\mid v \mbox{ in $\RT$}\}$, that we call the \emph{language} of $\RT$. If $\RT$ has $n$ nodes, then  $L_{\RT}$ contains at most $n$ distinct words, each of which has length at most $n$.
We call a rooted tree $\RT$ \emph{deterministic} when the edges from a node to its children are labeled by pairwise distinct letters.
Throughout the paper we will assume that all rooted trees are deterministic, which can be ensured without losing the generality thanks to the following
lemma.

\begin{lemma}
\label{lem:proper}
Given a rooted tree $\RT$ on $n$ nodes we can construct in $O(n)$ time a deterministic rooted tree $\RT'$ with the same set of corresponding words.
\end{lemma}

\begin{proof}
The depth of a node of $\RT$ is its distance from the root. 
We start with sorting, for every $d=1,2,..$, the set of nodes $S(d)$ at depth
$d$ according to the labels of the edges leading to their parents. This can be
done in $O(n)$ total time with counting sort.
Then, we construct $\RT'$ by processing $S(0),S(1),S(2),..$. Assuming that we have
already identified, for every node $u\in S(d)$, its corresponding node $f(u)$ of $\RT'$,
we need to construct and identify the nodes $f(u')$ for every $u'\in S(d+1)$.
We process all nodes $u'\in S(d+1)$ in groups corresponding to the same letter
$a$ on the edge leading to their parent (because of the initial sorting we already have these groups available).
Denoting by $u$ the parent of $u'$ in $\RT$, we check if $f(u)$ has been already accessed
while processing the group of $a$, and if so we set $f(u')$ to be the already created node of $\RT'$.
Otherwise, we create a new edge outgoing from $f(u)$ to a new node $v$ in $\RT'$ and labeled with $a$,
and set $f(u')$ to be $v$. To check if $f(u)$ has been already accessed while processing the current
group (and retrieve the corresponding $f(u')$ if this is the case) we simply allocate an array $A$
of size $n$ indexed by nodes of $\RT'$ identified by number from $\{1,2,\ldots,n\}$. For every
entry of $A$ we additionally store a timestamp denoting the most recent group for which the
corresponding entry has been modified, and increase the timestamp after having processed the current
group.
\qed
\end{proof}

One could also define the set of words corresponding to a rooted tree $\RT$ by considering a set of words from the root to every node $v$
(in the literature this is sometimes called a \emph{forward trie}, as opposed to a \emph{backward trie}, cf.~\cite{Ine19}).
In our context, this distinction is meaningless, as the obtained languages are the same up to reversing all the words.

We say that a word $aub$, with $a,b\in\Sigma$, is a minimal absent word of a rooted tree $\RT$ if $aub$ is not a factor of any word $\nodestring(v)$ in $L_{\RT}$ but there exist words $\nodestring(v_{1})$ and $\nodestring(v_{2})$ in $L_{\RT}$ (not necessarily distinct) such that $au$ is a factor of  $\nodestring(v_{1})$ and $ub$ is a factor of $\nodestring(v_{2})$. That is, the set $\MAW(\RT)$ of minimal absent words of $\RT$ is the set of minimal absent words of the factorial closure of the language $L_{\RT}$. 
Since any word of length $n$ can be transformed into a unary rooted tree with $n+1$ nodes, some of the properties of minimal absent words for usual words can be transferred to rooted trees. Indeed, rooted trees are a strict generalization of words.

For unrooted trees, the definition of minimal absent words is analogous: We identify an unrooted tree $\UT$ with the language of words $L(\UT)$ corresponding to all (concatenations of labels of) simple paths that can be read in $\UT$ from any of its nodes. 
The language $L(\UT)$  contains $O(n^{2})$ words, each of which has length at most $n$. We therefore  define the set $\MAW(\UT)$ of
minimal absent words of $\UT$ as the set of  minimal absent words of the language $L(\UT)$, which in this case is already closed under
taking factors by definition. 

\paragraph{Our results. }
We prove that for any rooted tree with $n$ nodes there are $O(n\sigma)$ minimal absent words, and we show that this bound is tight.
For unrooted trees, we prove that the previous bound becomes $O(n^{2}\sigma)$, and we give an explicit construction that achieves this bound. We also consider the case of minimal absent words of fixed length and generalize a previously-known construction.

Furthermore, we present an algorithm that computes all the minimal absent words in a rooted tree $\RT$ with $n$ nodes in  output-sensitive time $O(n+|\MAW(\RT)|)$. This also yields an algorithm that computes all the minimal absent words in an unrooted tree $\UT$ with $n$ nodes in time $O(n^{2}+|\MAW(\UT)|)$. Note that while it is plausible that an efficient algorithm could have been designed, as in the case of words,
from a DAWG~\cite{DBLP:conf/mfcs/FujishigeTIBT16}, the size of the DAWG of a backward/forward tree is superlinear~\cite{Ine19},
so it is not immediately clear if such an approach would lead to an optimal algorithm.
Excluding the space necessary to store all the results, our algorithms need $O(n)$ and $O(n^{2})$ space, respectively.

Our algorithms are designed in the word-RAM model with $\Omega(\log n)$-bit words.

\section{Bounds on the number of minimal absent words}

Let $\RT$ be a rooted tree with $n$ nodes and edges labeled by letters from an integer alphabet $\Sigma$ of cardinality $\sigma$ polynomial
in $n$. Let the
language of $\RT$ be $L_{\RT}=\{\nodestring(v)\mid v \mbox{ in $\RT$}\}$, where $\nodestring(v)$ is the sequence of edge labels
from node $v$ to the root.

For convenience, we add a new root to $\RT$ and an edge labeled by a new letter $\$$ not belonging to $\Sigma$ from the new root to the old root. This corresponds to appending  $\$$ at the end of each word of $L_{\RT}$. We then arrange all the words of $L_{\RT}$ in a trie. Each node $u$ of this trie corresponds to a word obtained by concatenating the edges from the root of the trie to node $u$, so in this paper we will implicitly identify a node of the trie with the corresponding word in the set of prefixes of $L_{\RT}$.
Following a standard approach, if we compact this trie by collapsing maximal chains of edges with every inner node having exactly one child and edges labeled by words,  we obtain the suffix tree  $\ST$ of $\RT$. The nodes in $\ST$ (the branching nodes) are called explicit nodes, while the nodes of the trie that have been collapsed (the non-branching nodes) are called implicit. Because $\$$ does not belong to the original alphabet,
the leaves of $\ST$ are in one-to-one correspondence with the nodes of $\RT$.

A word $aub$, with $a,b\in \Sigma$, is a minimal absent word for $\RT$ if it is a minimal absent word for the factorial closure of $L_{\RT}$, that is, if both $au$ and $ub$ but not $aub$ are factors of some words in $L_{\RT}$. The set of minimal absent words of $\RT$ is denoted by $\MAW(\RT)$.

If $aub\in \MAW(\RT)$, then $au$ occurs as a factor in some word of $L_{\RT}$ but never followed by letter $b$, hence there exists a letter $b'\in \Sigma\cup\{\$\}$ such that $ub$ and $ub'$ can be read in $\ST$ spelled from the root (possibly terminating in an implicit node).
This implies that $u$ corresponds to an explicit node in $\ST$, and $b$ is the first letter on its outgoing edge. Consequently,
$ub$ can be identified with an edge of $\ST$, so the number of minimal absent words of $\RT$ is upper-bounded by the
product of $\sigma$ and the number of edges of $\ST$.

\begin{theorem}
The number of  minimal absent words of a rooted tree with $n$ nodes whose edges are labeled by letters from an alphabet of size $\sigma$ is $O(n\sigma)$.
\end{theorem}

Therefore, the same upper bound that holds for words also holds for rooted trees. As a consequence, we have that all known upper bounds for words, and constructions that realize them, are still valid for rooted trees. 

In particular, one question that has been studied is whether the upper bound $O(n\sigma)$ is still tight when one considers minimal absent words of a fixed length. Almirantis et al.~\cite[Lemma 2]{DBLP:journals/almob/AlmirantisCGIMP17} showed that the upper bound $O(n\sigma)$ for a fixed length of minimal absent words is tight if $\sqrt{n} < \sigma \leq n$. Actually, they showed that it is possible to construct words of any
length $n$, with $\sigma \leq n \leq \sigma (\sigma-1)$, having $\Omega(n\sigma)$ minimal absent words of length $3$.
We now give a construction that generalizes this result.  

Let $\Sigma=\{1,2,\ldots, \sigma\}$. For every $n$, let $k>1$ be such that $\sigma^{k} \leq n < \sigma^{k+1}$.  Let $\Sigma^{k}=\{s_{1},s_{2},\ldots ,s_{\sigma^{k}}\}$. For every $1\leq i\leq \sigma^{k}$ we define the word 
$$w_{i}=\$1s_{i}\$s_{i}1\$2s_{i}\$s_{i}2\$\cdots \$\sigma s_{i} \$ s_{i}\sigma\$,$$ 
where $\$$ is a new symbol not belonging to $\Sigma$. The length of each word $w_{i}$ is $2\sigma(k+2)+1$, which is smaller than $n$ up to excluding small cases \footnote{The reader may verify that for $k=2$, $|w_{i}|\leq \sigma^{k}$ as soon as $\sigma\geq 9$; for $k>2$,  $|w_{i}|\leq \sigma^{k}$ as soon as $\sigma+k\geq 7$.}. 

Let $\ell=\lfloor{n/|w_{i}|}\rfloor$ and set $w=w_{1}w_{2}\cdots w_{\ell}$, so that $|w|>n/2$.  
We have that $as_{i}b$ is a minimal absent word of $w$ for every $a,b\in \Sigma$ and $1\leq i\leq \ell$. So, $w$ has length $\Theta(k\sigma\ell)$ and there are $\Theta(\sigma^{2}\ell)$ minimal absent words of $w$ of length $k+2$. 

Thus, we have proved the following theorem.

\begin{theorem}\label{th:bound}
A word of length $n$ over an alphabet of size $\sigma$ can have $\Omega(n\sigma/\log_{\sigma}n)$ minimal absent words all of the same length.
\end{theorem}

Observe that for $\sqrt{n} < \sigma \leq n$, $\log_{\sigma}n=\Theta(1)$, therefore~\cref{th:bound} strictly generalizes Almirantis et al.'s result.

\medskip 
Let now $\UT$ be an unrooted tree. The number of distinct simple paths in $\UT$ is $O(n^{2})$. Since each minimal absent word $aub$ is uniquely described by a pair $(au,b)$ such that $au$ is a simple path in $\UT$ and $b$ is a letter, we have that the  number of minimal absent words of $\UT$ is upper-bounded by $O(n^{2}\sigma)$.

\begin{theorem}
The number of  minimal absent words of an unrooted tree with $n$ nodes whose edges are labeled by letters from an alphabet of size $\sigma$ is $O(n^{2}\sigma)$.
\end{theorem}

We now provide an example of an unrooted tree realizing this bound. Let $\Sigma=\{0,1,\ldots,s\}.$ Our unrooted tree $\UT$ is built as follows:
\begin{itemize}
\item  We first build a sequence of $N+1$ nodes such that every other node is connected to $s$ terminal nodes with edges labeled by $1,2,\ldots,s$ and is connected to the next node of the sequence with an edge labeled by $0$;
\item Then, we attach to each of the last nodes of the previous sequence  $s$ simple paths composed of $N$ nodes with edges labeled by $0$.
\end{itemize}
See Figure~\ref{fig:ut2} for an illustration.

\begin{figure}
  \begin{tikzpicture}  [scale=.75,auto=left]
  
    \node[draw,circle,fill,inner sep=1pt] (A-1) at (0,0) {};
    \node[draw,circle,fill,inner sep=1pt] (A-2) at (1.7,0) {};
    \node[draw,circle,fill,inner sep=1pt] (A-3) at (3.4,0) {};
    \node[draw,circle,fill,inner sep=1pt] (A-4) at (5.,0) {};
    \node[draw,circle,fill,inner sep=1pt] (A-5) at (6.7,0) {};
    \path[draw] (A-1) -- (A-2) node[pos=0.5,anchor=south] {$0$};
    \path[draw] (A-2) -- (A-3) node[pos=0.5,anchor=south] {$0$};
    \path[draw,dotted] (A-3) -- (A-4);
    \path[draw] (A-4) -- (A-5) node[pos=0.5,anchor=south] {$0$};
    \path[draw] (A-1) -- ++(-120:1.5) node[pos=0.5,anchor=east] {1} node[draw,circle,fill,inner sep=1pt] {};
    \path[draw] (A-1) -- ++(-110:1.5) node[pos=0.5,anchor=north] {2} node[draw,circle,fill,inner sep=1pt] {};
    \path[draw,dotted] (A-1) -- ++(-90:1.5);
    \path[draw] (A-1) -- ++(-70:1.5) node[pos=0.5,anchor=north] {$s$} node[draw,circle,fill,inner sep=1pt] {};
     \path[draw] (A-2) -- ++(-120:1.5) node[pos=0.5,anchor=east] {1} node[draw,circle,fill,inner sep=1pt] {};
    \path[draw] (A-2) -- ++(-110:1.5) node[pos=0.5,anchor=north] {2} node[draw,circle,fill,inner sep=1pt] {};
    \path[draw,dotted] (A-2) -- ++(-90:1.5);
    \path[draw] (A-2) -- ++(-70:1.5) node[pos=0.5,anchor=north] {$s$} node[draw,circle,fill,inner sep=1pt] {};
            \path[draw] (A-3) -- ++(-120:1.5) node[pos=0.5,anchor=east] {1} node[draw,circle,fill,inner sep=1pt] {};
    \path[draw] (A-3) -- ++(-110:1.5) node[pos=0.5,anchor=north] {2} node[draw,circle,fill,inner sep=1pt] {};
    \path[draw,dotted] (A-3) -- ++(-90:1.5);
    \path[draw] (A-3) -- ++(-70:1.5) node[pos=0.5,anchor=north] {$s$} node[draw,circle,fill,inner sep=1pt] {};
     \path[draw] (A-4) -- ++(-120:1.5) node[pos=0.5,anchor=east] {1} node[draw,circle,fill,inner sep=1pt] {};
    \path[draw] (A-4) -- ++(-110:1.5) node[pos=0.5,anchor=north] {2} node[draw,circle,fill,inner sep=1pt] {};
    \path[draw,dotted] (A-4) -- ++(-90:1.5);
    \path[draw] (A-4) -- ++(-70:1.5) node[pos=0.5,anchor=north] {$s$} node[draw,circle,fill,inner sep=1pt] {};
     \node[draw,circle,fill,inner sep=1pt] (A-6) at (8.2,1.5) {};
    \node[draw,circle,fill,inner sep=1pt] (A-7) at (9.9,1.5) {};
    \node[draw,circle,fill,inner sep=1pt] (A-8) at (11.6,1.5) {};
    \node[draw,circle,fill,inner sep=1pt] (A-9) at (13.2,1.5) {};
    \node[draw,circle,fill,inner sep=1pt] (A-10) at (14.9,1.5) {};
    \path[draw] (A-6) -- (A-7) node[pos=0.5,anchor=south] {0};
    \path[draw] (A-7) -- (A-8) node[pos=0.5,anchor=south] {0};
    \path[draw,dotted] (A-8) -- (A-9);
    \path[draw] (A-9) -- (A-10) node[pos=0.5,anchor=south] {$0$};
    \node[draw,circle,fill,inner sep=1pt] (A-11) at (8.2,.3) {};
    \node[draw,circle,fill,inner sep=1pt] (A-12) at (9.9,.3) {};
    \node[draw,circle,fill,inner sep=1pt] (A-13) at (11.6,.3) {};
    \node[draw,circle,fill,inner sep=1pt] (A-14) at (13.2,.3) {};
    \node[circle,fill,inner sep=1pt] (A-15) at (14.9,.3) {};
    \path[draw] (A-11) -- (A-12) node[pos=0.5,anchor=south] {0};
    \path[draw] (A-12) -- (A-13) node[pos=0.5,anchor=south] {0};
    \path[draw,dotted] (A-13) -- (A-14);
    \path[draw] (A-14) -- (A-15) node[pos=0.5,anchor=south] {$0$};
     \node[draw,circle,fill,inner sep=1pt] (A-16) at (8.2,-1.9) {};
    \node[draw,circle,fill,inner sep=1pt] (A-17) at (9.9,-1.9) {};
    \node[draw,circle,fill,inner sep=1pt] (A-18) at (11.6,-1.9) {};
    \node[draw,circle,fill,inner sep=1pt] (A-19) at (13.2,-1.9) {};
    \node[draw,circle,fill,inner sep=1pt] (A-20) at (14.9,-1.9) {};
    \path[draw] (A-16) -- (A-17) node[pos=0.5,anchor=north] {0};
    \path[draw] (A-17) -- (A-18) node[pos=0.5,anchor=north] {0};
    \path[draw,dotted] (A-18) -- (A-19);
    \path[draw] (A-19) -- (A-20) node[pos=0.5,anchor=north] {0};
    
    \path[draw] (A-5) -- (A-6) node[pos=0.5,anchor=east] {1};
    \path[draw] (A-5) -- (A-11) node[pos=0.65,anchor=south] {2};
    \path[draw] (A-5) -- (A-16) node[pos=0.5,anchor=north] {$s$};
    
    \node[draw,circle,fill,inner sep=1pt] (A-21) at (8.2,-.8) {};
	\path[draw,dotted] (A-5) -- (A-21);
	\node[draw,circle,fill,inner sep=1pt] (A-22) at (9.9,-.8) {};
    \node[draw,circle,fill,inner sep=1pt] (A-23) at (11.6,-.8) {};
    \node[draw,circle,fill,inner sep=1pt] (A-24) at (13.2,-.8) {};
    \node[draw,circle,fill,inner sep=1pt] (A-25) at (14.9,-.8) {};
    \path[draw] (A-21) -- (A-22) node[pos=0.5,anchor=north] {0};
    \path[draw] (A-22) -- (A-23) node[pos=0.5,anchor=north] {0};
    \path[draw,dotted] (A-23) -- (A-24);
    \path[draw] (A-24) -- (A-25) node[pos=0.5,anchor=north] {0};
	
	\draw[decorate,decoration={brace,amplitude=3pt,mirror}] 
    (-0.1,-3.5) coordinate (A1) -- (6.8,-3.5) coordinate (A-5); 
    \node at (3.3,-4){$N+1$};
    
    \draw[decorate,decoration={brace,amplitude=3pt,mirror}] 
    (8.1,-3.5) coordinate (A1) -- (15.0,-3.5) coordinate (A-5); 
    \node at (11.5,-4){$N+1$};
  \end{tikzpicture}

\caption{\label{fig:ut2} An unrooted tree realizing the upper bound on the number of minimal absent words.}
\end{figure}
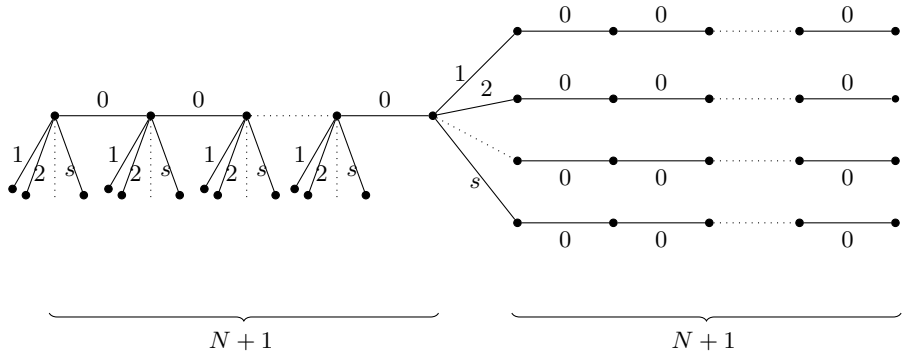

In total, $\UT$ has $(s+1)N+s(N+1)+1$ nodes. We therefore set $n=(s+1)N+s(N+1)+1$, so that $n =\Theta(sN)$. 

It is readily verified that for every $a,b,c$ in $\Sigma\setminus\{0\}$ and for every $0< j,k \leq  N$, there is a minimal absent word of the form $a0^{j}b0^{k}c$ (the prefix $a0^{j}b0^{k}$ can be found reading from the left part to the right part of the figure, while the suffix $0^{j}b0^{k}c$ can be found reading from the right part to the left part, the letter $b$ being one of the labels of the edges joining the left and the right part). 
Hence, the number of minimal absent words of $\UT$ is $\Omega(s^{3}N^{2})=\Omega(n^{2}\sigma)$. 

\begin{remark}
The previous construction can be modified in such a way that edges adjacent to a node have distinct labels, keeping the same bound on the number of minimal absent words. 
\end{remark}

\section{Algorithms for computing minimal absent words}

We now present an algorithm that computes the set $\MAW(\RT)$ of all minimal absent words of a rooted tree $\RT$ with $n$ nodes in output-sensitive time $O(n+|\MAW(\RT)|)$. 

We construct the suffix tree $\ST$ of $\RT$ in time $O(n)$~\cite{Shibuya03}. Recall that the leaves of $\ST$ are
in one-to-one correspondence with the nodes of $\RT$ and we can assume that every node $u$ of $\RT$ stores a pointer
to the leaf of $\ST$ corresponding to $\nodestring(u)$.

\begin{definition}
For every (implicit or explicit) node $u$ of $\ST$, we define the set $A(u)$ as the set of all letters $a\in \Sigma$ such that $au$ can be spelled from the root of $\ST$, i.e., there exists a node $v$ of $\RT$ such that $\nodestring(v)=auz$ for some (possibly empty) word $z$. 
\end{definition}

As already noted before, if $aub$ is a minimal absent word of $\RT$, then $au$  occurs as a factor in some word of $L_{\RT}$ followed by a letter $b'\in\Sigma\cup\{\$\}$ different from $b$, hence $u$ is an explicit node of $\ST$. 

\begin{lemma}\label{lem:a}
Let $u$ be an explicit node of the suffix tree $\ST$ of the tree $\RT$. Let $u_{1},u_{2},\ldots,u_{k}$ be the children of $u$ in the non-compacted trie from which we obtained $\ST$, and let $b_{1},b_{2},\ldots, b_{k}$ be the labels of the corresponding edges. Then, for every $1\leq i\leq k$ and every letter $$a_{j}\in (A(u_{1})\cup \ldots \cup A(u_{k}))\setminus A(u_{i}),$$ the word $a_{j}ub_{i}$ is a minimal absent word of $\RT$.

Conversely, every minimal absent word of $\RT$ is of the form $a_{j}ub_{i}$ described above.
\end{lemma}

\begin{proof}
Since $a_{j}$ does not belong to $A(u_{i})$, then by definition the word $a_{j}ub_{i}$ does not belong to $L_{\RT}$, but there exists $\ell\neq i$ such that $a_{j}\in A(u_{\ell})$, that is, $a_{j}ub_{\ell}$ is a factor of a word in $L_{\RT}$. Hence, $a_{j}u$ is a factor of a word in $L_{\RT}$. Since $ub_{i}$ is also a factor of a word in $L_{\RT}$ by construction, we have that $a_{j}ub_{i}$ is a minimal absent word of $\RT$.

Conversely, if $a_{j}ub_{i}$ is a minimal absent word of $\RT$, then $u$ occurs as a factor in some word of $L_{\RT}$ followed by different letters in $\Sigma\cup\{\$\}$, hence it corresponds to an explicit node in $\ST$, so all minimal absent words of $\RT$ are found in this way.
\qed
\end{proof}

\begin{definition}
For every leaf $u$ of $\ST$ we define the set $B(u)$ as the set of all letters $a\in \Sigma$ such that $au=\nodestring(v)$ for some node $v$ in $\RT$. 
\end{definition}

\begin{lemma}
\label{lem:b}
For every (implicit or explicit) node $u$ of $\ST$, we have $A(u)=\bigcup\{B(u')\mid u' \mbox{ is a leaf in the subtree of $\ST$ rooted at $u$}\}$.
\end{lemma}

\begin{proof}
Let $u'$ be a leaf in the subtree of $\ST$ rooted at $u$. Thus, the word $u$ is a prefix of the word $u'$,
i.e., $u'=uz$ for some word $z$. By definition, $B(u')$  is the set of all letters $a\in \Sigma$  such that $au'=\nodestring(v')$ for
some node $v'$ in $\RT$. That is, the set of all letters $a\in \Sigma$  such that $au'=auz$ is a word in $L_{\RT}$. On the other
hand, by definition, $A(u)$ is the set of all letters $a\in \Sigma$ such that $\nodestring(v)=auz$ for some node $v$ of $\RT$ and
word $z$. That is, the set of all letters $a\in \Sigma$  such that $auz$ is a word in $L_{\RT}$ for some word $z$.
\qed
\end{proof}

We now show how to compute, in time proportional to the output size, the set $\MAW(\RT)$.

We start with creating, for every letter $a \in \Sigma$, a list $L(a)$ of all leaves $u$ such that $a\in B(u)$ sorted in preorder.
The lists can be obtained in linear time by traversing all the non-root nodes $v\in \RT$, following the edge
labeled by $a$ from $v$ to its parent $v'$, and finally following the pointer from $v'$ to the leaf $v''$
of $\ST$ corresponding to $\nodestring(v')$ and adding $v''$ to $L(a)$. Finally, because the preorder numbers are from
$[n]$ the lists can be sorted in linear time with counting sort. 

Now we iterate over all letters $a\in \Sigma$. Due to \cref{lem:a}, the goal is to extract all explicit nodes $u\in \ST$
such that, for some child $u_{i}$ of $u$ such that the $b_{i}$ is the first letter on the edge from $u$ to $u_{i}$,
$a u b_{i}$ is a minimal absent word. By~\cref{lem:b}, this is equivalent to $u$ having a descendant $u'\in L(a)$
(where possibly $u=u'$) and $u_{i}$ not having any such descendant. This suggests that we should work
with the subtree of $\ST$, denoted $\ST(a)$, induced by all leaves $v\in L(a)$. Formally, $u\in \ST(a)$ when
$u'\in L(a)$ for some leaf $u'$ in the subtree of $u$. Even though $\ST$ does not contain nodes with just one child,
this is no longer the case for $\ST(a)$. Thus, we actually work with its compact version, denoted $\ST(a)$. 
Every node of $\ST(a)$ stores a pointer to its corresponding node of $\ST$.
Assuming that $\ST$ has been preprocessed for constant-time Lowest Common Ancestor queries (which can be done in linear time and space~\cite{DBLP:journals/siamcomp/SchieberV88,BF00}), we can construct $\ST(a)$ efficiently due to the following lemma.

\begin{lemma}
\label{lem:induced}
Given $L(a)$, we can construct $\ST(a)$ in $O(|L(a)|)$ time.
\end{lemma}

\begin{proof}
The procedure follows the general idea used in the well-known linear time procedure for creating a Cartesian tree~\cite{Gabow84}.
We process the nodes $u\in L(a)$ in preorder and maintain a compact version of the subtree of $\ST$ induced by
all the already-processed nodes. Additionally, we maintain a stack storing the edges on its rightmost path.
Processing $u\in L(a)$ requires popping a number of edges from the stack, possibly splitting the topmost edge
into two (with one immediately popped as well), and pushing a new edge ending at $u$. Checking if an edge
should be popped, and also determining if (and how) an edge should be split, can be implemented with LCA
queries on $\ST$, assuming that we maintain pointers to the corresponding nodes of $\ST$.
\qed
\end{proof}

Having constructed $\ST(a)$, we need to consider two cases corresponding to $u$ being an explicit or an implicit node
of $\ST(a)$. In the former case, we need to extract the edges outgoing from $u$ in $\ST$ such that there is no edge outgoing
from the corresponding node in $\ST(a)$ starting with the same letter $b$, and output $aub$ as a minimal absent word.
Assuming that the outgoing edges are sorted
by their first letters, this can be easily done in time proportional to the degree of $u$ plus the number of extracted letters.
In the latter case, let the implicit node belong to an edge connecting $u$ to $v$ in $\ST(a)$, and let $u'$ and $v'$
be their corresponding nodes in $\ST$ with $u'$ being an ancestor of $v'$. We iterate through all explicit nodes
between $u'$ and $v'$ in $\ST$ and  for each such node we extract all of its outgoing edges. 
For each such edge
we check if $v'$ belongs to the subtree rooted at its endpoint other than $u$, and if not, extract its first letter $b$
to output $aub$ as a minimal absent word.

The overall time for every letter $a\in\Sigma$ can be bounded by the sum of the size of $\ST(a)$ and the number of generated
minimal absent words. Because $\sum_{a\in\Sigma}|L(a)| = O(n)$ and the size of $\ST(a)$ can be bounded by $O(|L(a)|)$,
the total time complexity is $O(n+|\MAW(\RT)|)$.

\medskip 

The previous algorithm can be used to design an algorithm that outputs all the minimal absent words of an unrooted tree $\UT$ with $n$ nodes in time $O(n^{2}+|\MAW(\UT)|)$ as follows.
For every node $u$ of $\UT$, we create a rooted tree $\RT_{u}$ by fixing $u$ as the root.
Then we merge all trees $\RT_{u}$ into a single tree $\RT$ of size $O(n^{2})$ by identifying their roots. 
Finally, we apply Lemma~\ref{lem:proper} to make $\RT$ deterministic and apply our algorithm for rooted trees
in $O(n^{2})$ total time.

\section{Acknowledgments}

 This research was carried out during a visit of the first author to the Institute of Computer Science of the University of Wroc{\l}aw, supported by grant CORI-2018-D-D11-010133 of the University of Palermo. The first author is also supported by MIUR project PRIN 2017K7XPAN ``Algorithms, Data Structures and Combinatorics for Machine Learning''.
 
 We thank anonymous reviewers for helpful comments.


\begin{thebibliography}{10}

\bibitem{DBLP:journals/almob/AlmirantisCGIMP17}
Y.~Almirantis, P.~Charalampopoulos, J.~Gao, C.~S. Iliopoulos, M.~Mohamed, S.~P.
  Pissis, and D.~Polychronopoulos.
\newblock On avoided words, absent words, and their application to biological
  sequence analysis.
\newblock {\em Algorithms for Molecular Biology}, 12(1):5:1--5:12, 2017.

\bibitem{DBLP:journals/bmcbi/BartonHMP14}
C.~Barton, A.~H{\'{e}}liou, L.~Mouchard, and S.~P. Pissis.
\newblock Linear-time computation of minimal absent words using suffix array.
\newblock {\em {BMC} Bioinformatics}, 15:388, 2014.

\bibitem{DBLP:journals/fuin/BealCM03}
M.~B{\'{e}}al, M.~Crochemore, F.~Mignosi, A.~Restivo, and M.~Sciortino.
\newblock Computing forbidden words of regular languages.
\newblock {\em Fundam. Inform.}, 56(1-2):121--135, 2003.

\bibitem{DBLP:journals/ijac/BealFM05}
M.~B{\'{e}}al, F.~Fiorenzi, and F.~Mignosi.
\newblock Minimal forbidden patterns of multi-dimensional shifts.
\newblock {\em {IJAC}}, 15(1):73--93, 2005.

\bibitem{DBLP:conf/stacs/BealMR96}
M.~B{\'{e}}al, F.~Mignosi, and A.~Restivo.
\newblock Minimal forbidden words and symbolic dynamics.
\newblock In {\em {STACS}}, volume 1046 of {\em Lecture Notes in Computer
  Science}, pages 555--566. Springer, 1996.

\bibitem{DBLP:journals/algorithmica/BelazzouguiC17}
D.~Belazzougui and F.~Cunial.
\newblock A framework for space-efficient string kernels.
\newblock {\em Algorithmica}, 79(3):857--883, 2017.

\bibitem{DBLP:conf/esa/BelazzouguiCKM13}
D.~Belazzougui, F.~Cunial, J.~K{\"{a}}rkk{\"{a}}inen, and V.~M{\"{a}}kinen.
\newblock Versatile succinct representations of the bidirectional
  {Burrows-Wheeler} transform.
\newblock In H.~L. Bodlaender and G.~F. Italiano, editors, {\em Algorithms -
  {ESA} 2013 - 21st Annual European Symposium, Sophia Antipolis, France,
  September 2-4, 2013. Proceedings}, volume 8125 of {\em Lecture Notes in
  Computer Science}, pages 133--144. Springer, 2013.

\bibitem{BF00}
M.~A. Bender and M.~Farach-Colton.
\newblock The {LCA} problem revisited.
\newblock In G.~H. Gonnet and A.~Viola, editors, {\em LATIN 2000: Theoretical
  Informatics}, pages 88--94, Berlin, Heidelberg, 2000. Springer Berlin
  Heidelberg.

\bibitem{Chairungsee2012109}
S.~Chairungsee and M.~Crochemore.
\newblock Using minimal absent words to build phylogeny.
\newblock {\em Theoretical Computer Science}, 450:109--116, 2012.

\bibitem{Charalampopoulos2018}
P.~Charalampopoulos, M.~Crochemore, G.~Fici, R.~Mercaş, and S.~P. Pissis.
\newblock Alignment-free sequence comparison using absent words.
\newblock {\em Information and Computation}, 262(1):57--68, 2018.

\bibitem{SPIRE2018}
P.~Charalampopoulos, M.~Crochemore, and S.~P. Pissis.
\newblock On extended special factors of a word.
\newblock In T.~Gagie, A.~Moffat, G.~Navarro, and E.~Cuadros{-}Vargas, editors,
  {\em String Processing and Information Retrieval - 25th International
  Symposium, {SPIRE} 2018, Lima, Peru, October 9-11, 2018, Proceedings}, volume
  11147 of {\em Lecture Notes in Computer Science}, pages 131--138. Springer,
  2018.

\bibitem{DBLP:conf/ismir/CrawfordB018}
T.~Crawford, G.~Badkobeh, and D.~Lewis.
\newblock Searching page-images of early music scanned with {OMR:} {A} scalable
  solution using minimal absent words.
\newblock In {\em {ISMIR}}, pages 233--239, 2018.

\bibitem{DBLP:conf/fct/CrochemoreHKMPR17}
M.~Crochemore, A.~H{\'{e}}liou, G.~Kucherov, L.~Mouchard, S.~P. Pissis, and
  Y.~Ramusat.
\newblock Minimal absent words in a sliding window and applications to on-line
  pattern matching.
\newblock In R.~Klasing and M.~Zeitoun, editors, {\em Fundamentals of
  Computation Theory - 21st International Symposium, {FCT} 2017, Bordeaux,
  France, September 11-13, 2017, Proceedings}, volume 10472 of {\em Lecture
  Notes in Computer Science}, pages 164--176. Springer, 2017.

\bibitem{DBLP:journals/ipl/CrochemoreMR98}
M.~Crochemore, F.~Mignosi, and A.~Restivo.
\newblock Automata and forbidden words.
\newblock {\em Information Processing Letters}, 67(3):111--117, 1998.

\bibitem{DCA}
M.~Crochemore, F.~Mignosi, A.~Restivo, and S.~Salemi.
\newblock Data compression using antidictionaries.
\newblock {\em Proceedings of the IEEE}, 88(11):1756--1768, 2000.

\bibitem{DBLP:conf/sccc/CrochemoreN02}
M.~Crochemore and G.~Navarro.
\newblock Improved antidictionary based compression.
\newblock In {\em 22nd International Conference of the Chilean Computer Science
  Society {(SCCC} 2002), 6-8 November 2002, Copiapo, Chile}, pages 7--13, 2002.

\bibitem{DBLP:conf/dcc/FialaH08}
M.~Fiala and J.~Holub.
\newblock {DCA} using suffix arrays.
\newblock In {\em 2008 Data Compression Conference {(DCC} 2008), 25-27 March
  2008, Snowbird, UT, {USA}}, page 516. {IEEE} Computer Society, 2008.

\bibitem{fici}
G.~Fici.
\newblock {\em Minimal Forbidden Words and Applications}.
\newblock PhD thesis, Universit\'{e} de Marne-la-Vall\'{e}e, 2006.

\bibitem{FiLaLoRi18}
G.~Fici, A.~Langiu, G.~Lo~Bosco, and R.~Rizzo.
\newblock {B}acteria {C}lassification using {M}inimal {A}bsent {W}ords.
\newblock {\em AIMS Medical Science}, 5(1):23--32, 2018.

\bibitem{FICI2006214}
G.~Fici, F.~Mignosi, A.~Restivo, and M.~Sciortino.
\newblock Word assembly through minimal forbidden words.
\newblock {\em Theoretical Computer Science}, 359(1):214--230, 2006.

\bibitem{FiReRi19}
G.~Fici, A.~Restivo, and L.~Rizzo.
\newblock Minimal forbidden factors of circular words.
\newblock {\em Theoret. Comput. Sci.}, to appear.

\bibitem{DBLP:conf/mfcs/FujishigeTIBT16}
Y.~Fujishige, Y.~Tsujimaru, S.~Inenaga, H.~Bannai, and M.~Takeda.
\newblock Computing {DAWG}s and minimal absent words in linear time for integer
  alphabets.
\newblock In P.~Faliszewski, A.~Muscholl, and R.~Niedermeier, editors, {\em
  41st International Symposium on Mathematical Foundations of Computer Science,
  {MFCS} 2016, August 22-26, 2016 - Krak{\'{o}}w, Poland}, volume~58 of {\em
  LIPIcs}, pages 38:1--38:14. Schloss Dagstuhl - Leibniz-Zentrum fuer
  Informatik, 2016.

\bibitem{Gabow84}
H.~N. Gabow, J.~L. Bentley, and R.~E. Tarjan.
\newblock Scaling and related techniques for geometry problems.
\newblock In {\em Proceedings of the Sixteenth Annual ACM Symposium on Theory
  of Computing}, STOC '84, pages 135--143, New York, NY, USA, 1984. ACM.

\bibitem{Ine19}
S.~Inenaga.
\newblock Suffix {T}rees, {DAWG}s and {CDAWG}s for {F}orward and {B}ackward
  {T}ries.
\newblock {\em CoRR}, abs/1904.04513, 2019.

\bibitem{DBLP:conf/birthday/MignosiRS99}
F.~Mignosi, A.~Restivo, and M.~Sciortino.
\newblock Forbidden factors in finite and infinite words.
\newblock In {\em Jewels are Forever}, pages 339--350. Springer, 1999.

\bibitem{DBLP:journals/ita/MignosiRS01}
F.~Mignosi, A.~Restivo, and M.~Sciortino.
\newblock Forbidden factors and fragment assembly.
\newblock {\em {ITA}}, 35(6):565--577, 2001.

\bibitem{DBLP:journals/tcs/MignosiRS02}
F.~Mignosi, A.~Restivo, and M.~Sciortino.
\newblock Words and forbidden factors.
\newblock {\em Theor. Comput. Sci.}, 273(1-2):99--117, 2002.

\bibitem{DBLP:conf/isita/OtaM10}
T.~Ota and H.~Morita.
\newblock On the adaptive antidictionary code using minimal forbidden words
  with constant lengths.
\newblock In {\em Proceedings of the International Symposium on Information
  Theory and its Applications, {ISITA} 2010, 17-20 October 2010, Taichung,
  Taiwan}, pages 72--77. {IEEE}, 2010.

\bibitem{DBLP:journals/siamcomp/SchieberV88}
B.~Schieber and U.~Vishkin.
\newblock On finding lowest common ancestors: Simplification and
  parallelization.
\newblock {\em {SIAM} J. Comput.}, 17(6):1253--1262, 1988.

\bibitem{Shibuya03}
T.~Shibuya.
\newblock Constructing the suffix tree of a tree with a large alphabet.
\newblock {\em IEICE Transactions on Fundamentals of Electronics,
  Communications and Computer Sciences}, E86-A(5):1061--1066, 2003.

\bibitem{DBLP:journals/bioinformatics/SilvaPCPF15}
R.~M. Silva, D.~Pratas, L.~Castro, A.~J. Pinho, and P.~J. S.~G. Ferreira.
\newblock Three minimal sequences found in {Ebola} virus genomes and absent
  from human {DNA}.
\newblock {\em Bioinformatics}, 31(15):2421--2425, 2015.

\end{thebibliography}
\end{document}